\documentclass{article}
\usepackage{ijcai25}

\usepackage{times}
\usepackage{soul}
\usepackage{url}
\usepackage[hidelinks]{hyperref}
\usepackage[utf8]{inputenc}
\usepackage[small]{caption}
\usepackage{graphicx}
\usepackage{amsmath}
\usepackage{amsthm}
\usepackage{booktabs}
\usepackage{algorithm}
\usepackage[switch]{lineno}

\usepackage{mathtools}
\usepackage{pgfplots}
\usepackage{natbib}
\usepackage{xspace}
\usepackage{cleveref}
\usepackage{amsfonts}
\usepackage[noend]{algpseudocode}

\newtheorem{theorem}{Theorem}[section]

\newtheorem{proposition}[theorem]{Proposition}
\newtheorem{lemma}[theorem]{Lemma}

\theoremstyle{definition}
\newtheorem{definition}[theorem]{Definition}
\newtheorem{example}[theorem]{Example}

\renewcommand{\sc}{\textsc{sc}}

\DeclarePairedDelimiter{\set}{\{}{\}}
\DeclarePairedDelimiter\ceiling{\lceil}{\rceil} 
\DeclarePairedDelimiter\floor{\lfloor}{\rfloor} 

\newcommand{\FVRFull}{flexible-voter representation\xspace}
\newcommand{\FVR}{\text{FVR}\xspace}
\newcommand{\FVRMath}{\mathrm{FVR}\xspace}
\newcommand{\Ropt}{R_{\text{OPT}}\xspace}

\DeclareMathOperator*{\argmin}{arg\,min}
\DeclareMathOperator*{\E}{\mathbb{E}}

\allowdisplaybreaks

\title{The Proportional Veto Principle for Approval Ballots}

\author{
Daniel Halpern$^1$
\and
Ariel D. Procaccia$^1$\and
Warut Suksompong$^2$\\
\affiliations
$^1$Harvard University\\
$^2$National University of Singapore
}

\begin{document}

\maketitle

\begin{abstract}
The proportional veto principle, which captures the idea that a candidate vetoed by a large group of voters should not be chosen, has been studied for ranked ballots in single-winner voting.
We introduce a version of this principle for approval ballots, which we call flexible-voter representation (FVR).
We show that while the approval voting rule and other natural scoring rules provide the optimal FVR guarantee only for some flexibility threshold, there exists a scoring rule that is FVR-optimal for all thresholds simultaneously.
We also extend our results to multi-winner voting.
\end{abstract}

\section{Introduction}

Voting is one of the most prominent approaches to collective decision-making and has been studied extensively in several disciplines, not least in social choice theory \citep{BF02,Zwic16}.
Some applications of voting require selecting a single winner, e.g., the president of a country or the venue for an event.
Other applications involve choosing multiple winners, such as members of a council, places to visit on a trip, or recipients of a research grant.

In several voting scenarios, a natural consideration is that groups of voters should have the right to veto candidates that they dislike, be it a politician whose views they disagree with or a venue they find unsuitable.
This intuition is formalized by the concept of the \emph{proportional veto core}, which was proposed by \citet{Moul81} and further studied by a number of authors \citep{Moul82,IK21,KK23,Pete23,KI24}.
The proportional veto core is defined for single-winner voting when each voter expresses her preference over candidates in the form of a strict ranking.
According to this principle, a candidate $a$ can be vetoed by a group of voters $T$ if there exists a sufficiently large set of candidates~$B$ such that every voter in $T$ prefers all candidates in $B$ to $a$---the required size of~$B$ depends on the size of $T$.
A significant result by \citet{Moul81} states that the proportional veto core is always non-empty, i.e., there always exists a candidate that is not vetoed by any group of voters.

While ranked ballots have been examined since the dawn of social choice theory, an alternative method for eliciting voter preferences, which has also received substantial attention in the literature, is via \emph{approval ballots} \citep{BF07,LS10}.
In an approval ballot, a voter may either approve or disapprove each candidate.
In addition to their simplicity, approval ballots are quite expressive, as voters are allowed to approve as many or as few candidates as they wish.
As \citet[pp.~2--3]{BP22} discussed, approval ballots arise naturally when preferences are dichotomous---for example, when they correspond to whether workers are available during various time slots, or whether hiring committee members consider candidates to be capable of performing a clearly defined task.
Approval ballots have been used in political contexts---including in the St. Louis, Missouri mayoral and Fargo, North Dakota municipal elections---as well as by professional groups such as the American Mathematical Society.

A tempting approach for applying the proportional veto principle to approval ballots is to convert the approval ballots into ranked ballots by breaking ties arbitrarily, and apply the result of \citet{Moul81} to obtain a candidate in the proportional veto core of the ranked ballots.
However, this approach may lead to patently undesirable outcomes.
For instance, consider four candidates $a,b,c,d$ and three voters who approve $\{b,c\},\{b,d\},\{c,d\}$, respectively.
One can convert these approval ballots into the ranked ballots where voter~$1$ ranks $b\succ c\succ a\succ d$, voter~$2$ ranks $b\succ d\succ a\succ c$, and voter~$3$ ranks $c\succ d\succ a\succ b$.
For these approval ballots, voter~$1$ vetoes $d$ for being last, and similarly voters~$2$ and $3$ veto $c$ and $b$, respectively, leaving $a$ as the only candidate in the proportional veto core.
However, the choice of $a$ is clearly suboptimal with respect to the original approval ballots, as it is approved by none of the voters.

To formalize the proportional veto principle for approval ballots, we introduce a concept of voter \emph{flexibility}, which corresponds to the fraction of candidates that a voter approves.
Intuitively, in single-winner voting, if a group of voters is sufficiently large and each voter in the group is sufficiently flexible, we would like to ensure that at least one voter in the group approves the chosen candidate.
Put differently, a sufficiently large and sufficiently flexible group has the power to veto a candidate that the entire group disapproves. 
While considering flexibility is natural in any setting where approval voting is employed, it may be especially useful in applications such as participatory budgeting, where voters sometimes participate to support a single project that aligns with a particular cause.
Indeed, rules that reward flexibility could encourage voters to approve more projects, lead to chosen projects with broader consensus, and provide greater legitimacy to the process.\footnote{We also remark that proportionality is a central guarantee behind the \emph{method of equal shares}, which has been used for participatory budgeting in Poland and Switzerland \citep{BFJP+24}.}
As another example, consider a group of countries deciding the location(s) for hosting an international event.
In this scenario, flexibility could help motivate the representatives of each country to vote beyond merely their own country or close allies.
We refer to our principle as \emph{\FVRFull (\FVR)}.

We highlight that \FVR is more robust than proportional veto for ranked ballots.
In particular, for a candidate $a$ to be vetoed by a group $T$, \FVR does not require every voter in~$T$ to approve all candidates in some set $B$ and disapprove $a$---instead, it only requires all voters in~$T$ to approve sufficiently many (but not necessarily the same) candidates and disapprove $a$.
We discuss this distinction further in \Cref{sec:additional-discussion}, but note here that \FVR is particularly suited to approval ballots, where voter flexibility can be defined naturally. 

\subsection{Overview of Results}
\label{sec:overview}

In \Cref{sec:single-winner}, we focus on single-winner voting.
For $s\in (0,1)$, we say that a voter is \emph{$s$-flexible} if she approves at least an $s$-fraction of the candidates.
For each voting rule $R$ and each $s$, we denote by $\FVRMath(R,s)$ the smallest value of $r$ such that whenever a group of voters constitutes more than an $r$-fraction of the voters and each voter in the group is $s$-flexible, at least one voter in the group approves the candidate chosen by $R$.
The lower $\FVRMath(R,s)$ is, the stronger the guarantee that the rule~$R$ provides for the threshold~$s$.

We first derive a general lower bound: for any rule $R$ and any $s$, it holds that $\FVRMath(R,s) \ge 1-s$.
Moreover, this bound is tight in the sense that for each $s$, the rule $R_{s\text{-threshold}}$, which chooses a candidate with the most approvals among $s$-flexible voters, achieves $\FVRMath(R_{s\text{-threshold}},s) = 1-s$.
However, as each rule $R_{s\text{-threshold}}$ is tailored to a specific value of~$s$, a rule with strong guarantees for all values of~$s$ simultaneously would be more desirable.
We show that the classic approval voting rule, which chooses a candidate with the most approvals overall, yields $\FVRMath(R,s) = \frac{1}{1+s}$ for all~$s$; this essentially matches the lower bound for $s$ close to $0$, but becomes further away as $s$ increases.
To achieve better guarantees for larger $s$, we consider placing more weight on voters with higher flexibility.
Specifically, we let the contribution of a voter to the score of each of her approved candidates be the voter's flexibility raised to the $p$-th power, where $p > 0$ is a given parameter.
For each $p$, we determine the exact \FVR of the $p$-power scoring rule---we find that it matches the lower bound at $s = \frac{p}{1+p}$, but not at other values of $s$.

From the aforementioned results, one may be tempted to believe that there is an inevitable trade-off between optimizing for different thresholds $s$.
Strikingly, we show that this is not the case: there is in fact a rule $\Ropt$ that is \emph{\FVR-optimal}, meaning that $\FVR(R,s) = 1-s$ for all $s$ simultaneously.
One may interpret this result as demonstrating that an approval-ballot analog of the proportional veto core is always non-empty.
Like $p$-power scoring rules, $\Ropt$ can be described as a scoring rule, with the scores chosen carefully.
We also establish a general theorem that can be used to determine \FVR guarantees of arbitrary scoring rules---this theorem allows us to characterize $\Ropt$ as the unique scoring rule (up to scaling) that satisfies \FVR-optimality.

In \Cref{sec:multi-winner}, we turn our attention to the more general setting of multi-winner voting, where the goal is to choose a set of $k$ candidates (called a \emph{committee}) from the $m$ given candidates; single-winner voting corresponds to $k=1$.
For $1\le t\le k$, we say that a voter \emph{$t$-approves} a committee if she approves at least $t$ candidates in the committee.
Unlike in the single-winner setting, it will be more convenient to state our bounds including their dependency on $m$ (rather than in the worst case over all $m$).
Specifically, we let $\FVRMath(R,s,t,m)$ be the smallest value of $r$ such that whenever a group of voters constitutes more than an $r$-fraction of the voters, each voter is $s$-flexible, and there are $m$ candidates, then at least one voter in the group $t$-approves the committee chosen by~$R$.

Similarly to single-winner voting, we provide a lower bound on $\FVRMath(R,s,t,m)$, which can be written in terms of a hypergeometric random variable; for any $s,t,m,k$, we present a rule for which this bound is tight.
Moreover, if $k$ and $t$ are fixed, we show that there exists a rule that yields the optimal guarantee for all $m$ and $s$.
This rule operates by iteratively choosing a candidate using a scoring rule; however, the weight assigned to each voter in this procedure depends nontrivially on the voter's preference and the candidates added thus far.
On the other hand, we prove that if we fix $s$ (and $k,m$), it may not be possible to achieve \FVR-optimality for different values of $t$ at the same time. 
We also demonstrate the incompatibility between \FVR-optimality and other notions of representation in the multi-winner literature.

\subsection{Additional Discussion}
\label{sec:additional-discussion}

As mentioned earlier, in single-winner voting, \FVR can be viewed as an approval-ballot analog of the proportional veto core for ranked ballots \citep{Moul81}.
However, in order for a group of voters to veto a candidate~$a$, \FVR does not require the voters in the group to commonly approve a set of candidates and disapprove $a$---instead, it is enough that these voters approve sufficiently many candidates and disapprove~$a$.
To highlight the difference that this distinction makes, note that strengthening the proportional veto core along the lines of our \FVR definition may lead to an empty core.
Indeed, if there are two voters with rankings $a\succ b\succ c$ and $c\succ b\succ a$ over three candidates, then this more demanding notion would rule out $c$ (as the first voter ranks it last), $a$ (as the second voter ranks it last), and $b$ (as both voters rank it second-last), leaving no viable candidate.\footnote{We formalize this in \Cref{app:PVC}.
Note that the proportional veto core does not rule out $b$, since the two voters do not agree on a candidate that is preferred to $b$.}
In the approval setting, \FVR takes advantage of the dichotomous nature of the preferences to achieve a stronger guarantee.

Another related concept is \emph{justified representation (JR)} \citep{ABCE+17}, which has received significant attention recently in the context of multi-winner (approval) voting.
JR captures the idea that if a group of voters is sufficiently large and all voters in the group approve a common candidate, then at least one candidate approved by some voter in the group should be selected.
While \FVR bears some similarity to JR, a major difference is that \FVR does not require voters in the group to approve any candidate in common, but only that these voters be sufficiently flexible.
Note that in single-winner voting, JR does not provide any meaningful guarantee, and neither do its various strengthenings \citep{ABCE+17,SELF+17,PPS21,BP23}.

\section{Single-Winner Voting}
\label{sec:single-winner}

\subsection{Preliminaries}

For each positive integer $t$, let $[t] := \{1,2,\dots,t\}$.
There is a set $N$ of $n$ voters and a set $M$ of $m$ candidates. 
Each voter~$i$ has an approval set $A_i \subseteq M$. 
An instance consists of $N$, $M$, and $(A_i)_{i\in N}$.
For a candidate $a \in M$, let $N_a := \set{i \in N \mid a \in A_i}$ be the set of voters who approve it. 
For a voter $i\in N$, let $f_i := |A_i|/m$ be the fraction of candidates that $i$ approves.
In single-winner voting, a rule chooses a single candidate from each instance.

For $s\in (0,1)$, we say that a voter~$i$ is \emph{$s$-flexible} if $f_i \ge s$.
Given a rule~$R$, we are interested in, for each $s\in (0,1)$, the smallest $r\in [0,1]$ (as a function of $s$) such that whenever a group of strictly more than $rn$ voters (i.e., more than an $r$-fraction of the voters) are all $s$-flexible, at least one voter in the group approves the candidate chosen by the rule.\footnote{If $s = 1$, the answer is clearly $r = 0$ regardless of the rule, so we do not consider this trivial case.}
We refer to this condition as the \emph{\FVR condition}, and denote this value of $r$ by $\FVRMath(R,s)$.
Note that $\FVRMath(R,s)$ is always well-defined.
Indeed, $r = 1$ satisfies the \FVR condition vacuously; if some $r$ satisfies the condition then every $r' \ge r$ does as well; and if for some $r$ it holds that every $r' > r$ satisfies the condition, then $r$ also satisfies the condition (because the condition is phrased as ``strictly more than $rn$ voters'').
The smaller the value $\FVRMath(R,s)$, the stronger the \FVR guarantee provided by the rule~$R$ for the threshold~$s$.

\begin{example}
Consider a rule $R$ that always returns the first candidate, regardless of the voters' approvals.
We claim that $\FVRMath(R,s) = 1$ for all $s\in (0,1)$, i.e., this rule offers the worst possible \FVR guarantees.
Indeed, for any $s\in (0,1)$, consider an instance where each of the $n$ voters approves all $m$ candidates except the first candidate, where $m > \frac{1}{1-s}$.
All voters are $s$-flexible but none of them approves the candidate chosen by $R$, so the \FVR condition holds only for $r = 1$.
\end{example}

\subsection{Lower Bound}

We first derive a lower bound on the \FVR guarantee achievable by any rule.

\begin{theorem}
\label{thm:singlewinner-lowerbound}
For any rule $R$ and any $s\in (0,1)$, we have $\FVRMath(R,s) \ge 1-s$.
Moreover, for each $s\in (0,1)$, there exists a rule $R$ such that $\FVRMath(R,s) = 1-s$.
\end{theorem}

\begin{proof}
For the first statement, consider any rule $R$ and any $s\in (0,1)$, and take any $r < 1-s$.
It suffices to show that there exists an instance where a group of at least $rn$ voters are all $s$-flexible but do not approve the candidate chosen by $R$.
Consider an instance with sufficiently large $n$ and $m$ (to be made more precise later).
Assume that each of the $n$ voters approves exactly $\ceiling{sm}$ candidates (so all voters are $s$-flexible), and the voters distribute their approvals as equally as possible.\footnote{For example, we can let the voters distribute their approvals one voter after another. In each voter's turn, the voter approves $\ceiling{sm}$ candidates that have received the fewest approvals so far, breaking ties arbitrarily.}
Hence, each candidate receives at most $\left\lceil\frac{n\cdot\ceiling{sm}}{m}\right\rceil$ approvals.
Observe that
\begin{align*}
\left\lceil\frac{n\cdot\ceiling{sm}}{m}\right\rceil < \frac{n(sm+1)}{m}+1 = n\cdot\frac{sm+1}{m} + 1.
\end{align*}
As $m$ grows, $\frac{sm+1}{m}$ converges to $s$, so $n\cdot\frac{sm+1}{m} + 1$ converges to $sn+1$.
Since $s < 1-r$, we have $sn+1 < (1-r)n$ for sufficiently large $n$.
This means that for sufficiently large $n$ and $m$, each candidate is approved by at most $(1-r)n$ voters.
Consequently, no matter which candidate $R$ chooses, at least $rn$ voters disapprove it.

For the second statement, fix $s\in (0,1)$.
Consider the rule $R_{s\text{-threshold}}$ that chooses a candidate $a$ with the highest number of approvals among $s$-flexible voters, breaking ties arbitrarily; in particular, the rule ignores voters who are not $s$-flexible.
It suffices to show that whenever a group of more than $(1-s)n$ voters are all $s$-flexible, at least one voter in the group approves $a$.
Since each $s$-flexible voter approves at least an $s$-fraction of the candidates, the candidates are approved by at least an $s$-fraction of the $s$-flexible voters on average.
By definition of the rule $R_{s\text{-threshold}}$, the candidate $a$ is approved by at least an $s$-fraction of the $s$-flexible voters, which means that at most a $(1-s)$-fraction of the $s$-flexible voters disapprove $a$.
Since the number of $s$-flexible voters is at most $n$, this implies that at most $(1-s)n$ voters who are $s$-flexible disapprove $a$.
It follows that in any group of more than $(1-s)n$ voters who are $s$-flexible, at least one voter in the group approves $a$.
\end{proof}

While the rule $R_{s\text{-threshold}}$ in the proof of \Cref{thm:singlewinner-lowerbound} achieves the optimal \FVR guarantee for the corresponding flexibility threshold $s$, it is tailored to only one threshold $s$, which leads to unwanted effects---for instance, the rule completely ignores voters who are almost $s$-flexible.
As such, it would be more desirable to have a rule that yields good guarantees for several thresholds $s$ simultaneously.

\subsection{Approval Voting Rule}
\label{sec:approval}

We next analyze the classic \emph{approval voting rule}, which selects a candidate with the highest number of approvals, breaking ties arbitrarily.
For this rule, we determine the tight \FVR guarantee for every $s$.

\begin{theorem}
\label{thm:approval}
Let $R_{\emph{approval}}$ be the approval voting rule.
For each $s\in (0,1)$, we have $\FVRMath(R_{\emph{approval}},s) = \frac{1}{1+s}$.
\end{theorem}

\begin{proof}
Fix $s\in (0,1)$.
We first show the upper bound $\FVRMath(R_{\text{approval}},s) \le \frac{1}{1+s}$.
Take an arbitrary instance, and consider a group~$B$ of more than $\frac{1}{1+s}\cdot n$ voters, all of whom are $s$-flexible.
The total number of approvals made by the voters in $B$ is greater than $\frac{1}{1+s}\cdot n\cdot sm$.
Since there are $m$ candidates, there exists a candidate $a$ approved by more than $\frac{1}{1+s}\cdot n\cdot s = \frac{s}{1+s}\cdot n$ voters in $B$.
On the other hand, a candidate that is not approved by any voter in $B$ has an approval score of less than $(1 - \frac{1}{1+s})n = \frac{s}{1+s}\cdot n$.
In particular, such a candidate is approved by fewer voters than $a$.
Hence, the candidate chosen by the approval voting rule must be approved by some voter in $B$.

Next, we show that $\FVRMath(R_{\text{approval}},s) \ge \frac{1}{1+s}$.
Take any $r < \frac{1}{1+s}$.
It suffices to show that there exists an instance where a group of at least $rn$ voters are all $s$-flexible but do not approve the candidate chosen by the approval voting rule.
Consider an instance with sufficiently large $n$ and $m$ (to be made more precise later).
Let $C$ be a group of $\ceiling{rn}$ voters, and assume that all voters in $N\setminus C$ approve a common candidate $a$ (and no other candidate).
Moreover, assume that each voter in $C$ approves exactly $\ceiling{sm}$ candidates in $M\setminus\{a\}$, and these voters distribute their approvals among these candidates as equally as possible.
Note that $a$ receives $n - \ceiling{rn} = \floor{(1-r)n}$ approvals, while every other candidate receives at most $\left\lceil{\frac{\ceiling{rn}\cdot \ceiling{sm}}{m-1}}\right\rceil$ approvals.
Since $r < \frac{1}{1+s}$, we have $1-r > \frac{s}{1+s}$, and so $\floor{(1-r)n} > \frac{s}{1+s}\cdot n$ when $n$ is sufficiently large.
Hence, $a$ receives more than $\frac{s}{1+s}\cdot n$ approvals for sufficiently large $n$.
On the other hand, note that
\begin{align*}
\left\lceil{\frac{\ceiling{rn}\cdot \ceiling{sm}}{m-1}}\right\rceil
&< \frac{(rn+1)(sm+1)}{m-1} + 1 \\
&= (rn+1)\cdot \frac{sm+1}{m-1} + 1.
\end{align*}
As $m$ grows, $\frac{sm+1}{m-1}$ converges to $s$, so $(rn+1)\cdot \frac{sm+1}{m-1} + 1$ converges to $(rn+1)s + 1 = rns+s+1$.
Since $r < \frac{1}{1+s}$, we have $rns+s+1 < \frac{1}{1+s}\cdot ns = \frac{s}{1+s}\cdot n$ for sufficiently large $n$.
This means that for sufficiently large $n$ and $m$,  each candidate in $M\setminus\{a\}$ receives at most $\frac{s}{1+s}\cdot n$ approvals.
Consequently, the approval voting rule chooses the candidate $a$, which is approved by none of the voters in $C$.
\end{proof}

Comparing Theorems~\ref{thm:singlewinner-lowerbound} and \ref{thm:approval}, we find that the \FVR guarantee of the approval voting rule is essentially optimal for $s$ close to $0$, but gets further from optimal as $s$ increases.

\subsection{Power Scoring Rules}

To achieve better guarantees for larger $s$, an enticing idea is to give more weight to voters with higher flexibility.
However, unlike the rule $R_{s\text{-threshold}}$ in the proof of \Cref{thm:singlewinner-lowerbound} where this weight is binary (according to whether a voter is $s$-flexible), we shall use a smoother weighting scheme.

To this end, for each real number $p\ge 0$, we define the \emph{$p$-power approval score} of candidate $a$ as $\sc_p(a) := \sum_{i \in N_a} f_i^p$ (recall that $N_a$ denotes the set of voters who approve $a$).
We consider the \emph{$p$-power scoring rule}, which chooses a candidate that maximizes the $p$-power approval score, breaking ties arbitrarily.
Note that if $p = 0$, then $\sum_{i \in N_a} f_i^p = |N_a|$, so the rule reduces to the approval voting rule which we already analyzed in \Cref{thm:approval}.
As $p$ increases, the rule places more importance on voters with high flexibility.
For each $p$, we derive tight \FVR guarantees for the corresponding rule.

\begin{theorem}
\label{thm:power}
For each $p > 0$, let $R_{p\emph{-power}}$ be the $p$-power scoring rule.
For each $s\in (0,1)$, we have $\FVRMath(R_{p\emph{-power}},s) = \frac{1}{1 + \frac{(s(1+p))^{1+p}}{p^p}}$.
\end{theorem}

The proof of \Cref{thm:power}, along with all other missing proofs, can be found in \Cref{app:omitted}.
Comparing \Cref{thm:power} with \Cref{thm:singlewinner-lowerbound}, we find that the \FVR guarantee of the $p$-power scoring rule is optimal only at $s = \frac{p}{1+p}$.
Indeed, one can verify by algebraic manipulation that
\begin{align*}
\frac{1}{1 + \frac{(s(1+p))^{1+p}}{p^p}} = 1-s
\, \Longleftrightarrow \,
s^p(1-s) = \frac{p^p}{(1+p)^{1+p}}.
\end{align*}
Within the range $x\in [0,1]$, basic calculus shows that the function $x^p(1-x)$ is maximized at $x = \frac{p}{1+p}$, and the resulting maximum is $\frac{p^p}{(1+p)^{1+p}}$.
Hence, the only $s$ for which the \FVR guarantee of the $p$-power scoring rule is optimal is $s = \frac{p}{1+p}$, and the corresponding guarantee is $\frac{1}{1+p}$.

\begin{figure}
\begin{tikzpicture}
\begin{axis}[
    axis lines = left,
    xlabel = \(s\),
    ylabel = {\(\FVRMath(R_p,s)\)},
    xmin=0, xmax=1,
    ymin=0, ymax=1,
    xtick= {0,0.2,0.4,0.6,0.8,1},
    ytick= {0,0.2,0.4,0.6,0.8,1},
    legend pos=south west,
]
\addplot [
    domain=0:1, 
    samples=100, 
    color=black,
    style=very thick,
]
{1 - x};
\addlegendentry{Optimal}
\addplot [
    domain=0:1, 
    samples=100, 
    color=red,
    style=thick,
    ]
    {1/(1+x)};
\addlegendentry{\(p = 0\)}
\addplot [
    domain=0:1, 
    samples=100, 
    color=blue,
    style=dashed,
    style=thick,
    ]
    {1/(1+4*x^2)};
\addlegendentry{\(p = 1\)}
\addplot [
    domain=0:1, 
    samples=100, 
    color=brown,
    style=dotted,
    style=thick,
    ]
    {1/(1+27*x^3/4)};
\addlegendentry{\(p = 2\)}
\end{axis}
\end{tikzpicture}
\centering
\caption{\FVR guarantees of $p$-power scoring rules for flexibility thresholds $s\in (0,1)$, compared to the optimal guarantees indicated by the thick line.
For $p=1, 2$, the guarantees match the optimal ones at $s = 1/2, 2/3$, respectively.}
\label{fig:power}
\end{figure}
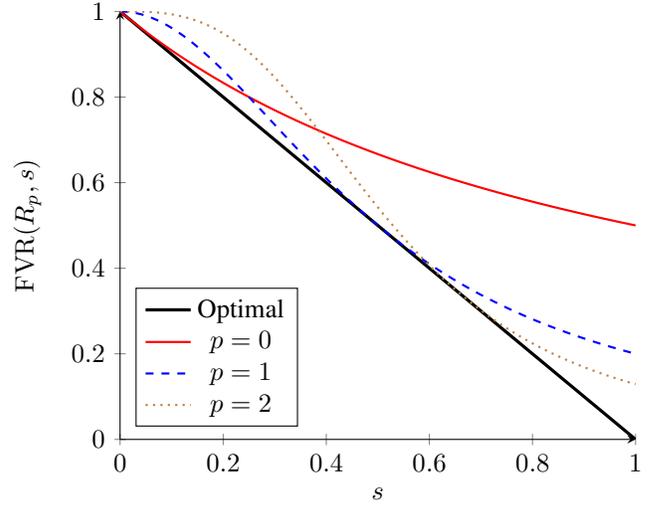

The \FVR guarantees of $p$-power scoring rules for $p = 0$ (i.e., approval voting rule), $1$, and $2$, as well as the optimal guarantees, are illustrated in \Cref{fig:power}.

\subsection{Optimal Scoring Rule}

Our discussion thus far raises an obvious question: is there an inherent trade-off between optimizing for different values of~$s$, or does there exist a rule that achieves the optimal \FVR guarantee for all $s$ simultaneously (that is, a rule whose guarantees exactly match the thick line in \Cref{fig:power})?
Perhaps surprisingly, we show that such a rule in fact exists, and moreover, comes from a simple class called \emph{scoring rules}~\citep{fishburn1979symmetric}.\footnote{In the original definition, these are called \emph{simple scoring rules}, but for conciseness, we drop the word `simple'.}

To define scoring rules, let a \emph{weight function} be a function $w: \mathbb{Q} \cap (0, 1) \to \mathbb{R}_{\ge 0}$ mapping each flexibility to a weight.
From this, we define the \emph{$w$-score} of a candidate $a$ to be $\sc^w(a) = \sum_{i \in N_a} w(f_i)$.
We assume without loss of generality that $w(1) = 0$; if a voter~$i$ has $f_i = 1$, then she approves all candidates, so modifying $w(1)$ changes the scores of all candidates by a constant. 
We call a weight function $w$  \emph{nontrivial} if $w(f) > 0$ for some $f \in (0, 1)$.\footnote{If $w(f) = 0$ for all $f$, then the $w$-scores of all candidates will be equal.}
We say that a rule is \emph{parameterized by $w$} if it always picks a candidate maximizing $\sc^w$, and we call it a \emph{scoring rule} if it is parameterized by some nontrivial $w$.   
Note that all of the rules we have considered so far are scoring rules: $R_{s\text{-threshold}}$ is parameterized by $w(f) = \mathbb{I}[f \ge s]$, $R_{\text{approval}}$ is parameterized by $w(f) = 1$, and $R_{p\text{-power}}$ is parameterized by $w(f) = f^p$.
 
For our optimal rule, unlike the power rules we have seen, we will choose the weight function $w(f) = \frac{1}{1-f}$.
Specifically, let $\Ropt$ be the rule that assigns a score of $\sum_{i\in N_a}\frac{1}{1-f_i}$ to each candidate~$a$ and selects a candidate with the highest score, breaking ties arbitrarily.\footnote{If a voter approves all candidates, we simply ignore that voter.}
We say that a rule~$R$ is \emph{\FVR-optimal} if $\FVRMath(R,s) = 1-s$ for all $s\in (0,1)$.

\begin{theorem}
\label{thm:singlewinner-optimal}
The rule $R_\emph{OPT}$ is \FVR-optimal.
\end{theorem}

\begin{proof}
Observe that an equivalent formulation of $\Ropt$ is that each voter~$i$ distributes $\frac{1}{1-f_i}$ points to each candidate that the voter \emph{disapproves}, and $\Ropt$ chooses a candidate with the \emph{lowest} total points.
In this formulation, each voter~$i$ distributes a total of $\frac{1}{1-f_i}\cdot (1-f_i)m = m$ points,\footnote{Again, we ignore voters who approve all candidates. Note that voters who approve none of the candidates also distribute points.} so at most $nm$ points are distributed overall. 
This means that the candidate~$a$ chosen by $\Ropt$ receives at most $n$ points.

Next, fix any $s\in(0,1)$, and consider a group of $s$-flexible voters who all disapprove $a$.
Since each voter in the group gives at least $\frac{1}{1-s}$ points to $a$, the size of the group is at most $(1-s)n$.
It follows that whenever a group of more than $(1-s)n$ voters are all $s$-flexible, at least one voter in the group approves $a$.
\end{proof}

\subsection{General Weight Function Analysis}
\label{subsec:general}
\Cref{thm:approval,thm:power,thm:singlewinner-optimal} provide analyses for specific scoring rules. 
What can we say about scoring rules more broadly?
In the following theorem, we analyze the \FVR guarantees of arbitrary scoring rules. 
Given a weight function $w$, let $R^w$ denote a scoring rule parameterized by $w$.

\begin{theorem}\label{thm:general-weight}
Let $w$ be a nontrivial weight function.
Then,
    \[
        \FVRMath(R^w, s) = \frac{\rho}{\rho + \varphi_s},
    \]
    where $\rho := \sup_{f} (1 - f)\cdot w(f)$ and $\varphi_s := \inf_{f \ge s} f \cdot w(f)$; if $\rho = \infty$, the ratio is defined to be $1$.
    Further, if $w$ is nondecreasing, then $\varphi_s = s\cdot w(s)$ and the bound simplifies to
    \[
        \FVRMath(R^w, s) = \frac{\rho}{\rho + s \cdot w(s)}.
    \]
\end{theorem}
\begin{proof}
    Fix $w$, $R^w$, and $s \in (0, 1)$. 
    We first bound $\FVR(R^w, s)$ from above. 
    If $\rho = \infty$ then the bound holds trivially, so assume that $\rho < \infty$.
    Fix an instance and let $a^*$ be the candidate chosen by $R^w$. 
Let $B$ be a set of $s$-flexible voters who disapprove $a^*$, and let $\bar{B} = N\setminus B$ be the set of remaining voters.
Note that  $\sc^{w}(a^*) \le \sum_{i \in \bar{B}} w(f_i)$, because no voter in $B$ can contribute to the score of $a^*$. 
Consider now the average score of all candidates, which is equal to
\begin{align*}
\frac{1}{m} \sum_{a \in M} \sc^{w}(a) &= \frac{1}{m} \sum_{a \in M} \sum_{i \in N_a} w(f_i) \\
&= \frac{1}{m} \sum_{i\in N} |A_i| \cdot w(f_i) = \sum_{i\in N} f_i\cdot w(f_i),
\end{align*}
where the second equality holds because each voter $i$ contributes $w(f_i)$ points to each of exactly $|A_i|$ candidates. 
Since $a^*$ has the largest $w$-score, it must have $w$-score at least the average.
This implies that
\[\sum_{i \in \bar{B}} w(f_i) - \sum_{i \in N} f_i \cdot w(f_i) \ge \sc^{w}(a^*) - \frac{1}{m} \sum_{a \in M} \sc^{w}(a) \ge 0.\]
On the other hand,
\begin{align*}
   &\sum_{i \in \bar{B}} w(f_i) - \sum_{i \in N} f_i \cdot w(f_i) \\
    &= \sum_{i \in \bar{B}} w(f_i) - \left(\sum_{i \in B}f_i \cdot w(f_i) + \sum_{i \in \bar{B}}f_i \cdot w(f_i)\right)\\
    &= \sum_{i \in \bar{B}} \left(1 - f_i\right)\cdot w(f_i) - \sum_{i \in B} f_i \cdot w(f_i)\\
    &\le (n - |B|) \rho - |B| \varphi_s,
\end{align*}
where the inequality holds by the definitions of $\rho$ and $\varphi_s$, as $f_i \ge s$ for $i \in B$. Finally, observe that
\begin{align*}
    (n - |B|) \rho - |B| \varphi_s \ge 0
    &\iff n \rho \ge |B|(\rho + \varphi_s)\\
    &\iff \frac{|B|}{n} \le \frac{\rho}{\rho + \varphi_s}.
\end{align*}
Note that since $w$ is nontrivial, we have $\rho > 0$, and so $\rho + \varphi_s > 0$ and the last transition is valid.

The remainder of the proof, which bounds $\FVR(R^w, s)$ from below, can be found in \Cref{app:general-weight}.
\qedhere
\end{proof}

\Cref{thm:general-weight} is a powerful tool that allows us to analyze arbitrary scoring rules. Indeed, \Cref{thm:approval,thm:power,thm:singlewinner-optimal} could be proven as consequences (although we keep them separate for exposition). 
In addition, we can derive other nontrivial consequences, including the following fact that $\Ropt$ is essentially the unique scoring rule that ensures \FVR-optimality. 

\begin{theorem}
\label{thm:characterization}
    Suppose that $R^w$ is \FVR-optimal for some nontrivial weight function $w$. 
    Then, there exists a positive real number $c$ such that $w(f) = \frac{c}{1 - f}$ for all $f$.
\end{theorem}

\section{Multi-Winner Voting}
\label{sec:multi-winner}

\subsection{Preliminaries}
In this section, we turn to the more general setting of multi-winner voting. 
Given an instance, our rules $R$ now choose a \emph{subset} of winning candidates of a given size $k$.
We refer to a subset of $k$ candidates as a \emph{$k$-committee} (or just \emph{committee} if $k$ is clear from the context) and refer to a rule that selects a $k$-committee as a \emph{$k$-committee rule}.
Assume without loss of generality that $k < m$, since if $k = m$, there is a single $k$-committee which also makes all voters maximally happy.

To generalize the \FVR condition to the multi-winner setting, we need to generalize the notion of a voter ``approving the chosen candidate'' to ``approving the chosen committee.'' 
Two natural generalizations immediately come to mind. 
Fix a committee $W$. 
A quite lenient notion we could consider is to say that a voter $i$ approves $W$ if she approves \emph{at least one} of the committee members (i.e., $W \cap A_i \ne \emptyset$), while a much more stringent requirement is that $i$ approves $W$ only if she approves \emph{all} of the committee members (i.e., $W \subseteq A_i$). 
We will generalize both of these notions at once by introducing an additional (positive integer) parameter $t \le k$, and saying that a voter \emph{$t$-approves} a committee $W$ if she approves at least $t$ of the committee members (i.e., $|W \cap A_i| \ge t$). 
If a voter does not $t$-approve a committee (i.e., $|W \cap A_i| < t$), we will say that she \emph{$t$-disapproves} it. The earlier notions correspond to the special cases of $t=1$ and $t=k$, respectively.

With this generalization in hand, we can now extend the notion of \FVR to the multi-winner setting. 
Unlike in the single-winner setting, it will be more convenient to state our guarantees including their dependency on $m$ (rather than in the worst case over all $m$), although we will primarily be interested in the behavior of these bounds as $m$ grows large. 
This is because in some cases, unlike in the single-winner setting, as $m$ grows, the worst case may occur for a fixed finite $m$ rather than in the limit. 
With this in mind, for a $k$-committee rule $R$, we define $\FVRMath(R, s, t, m)$ as the largest possible fraction of voters that are all $s$-flexible, yet $t$-disapprove the chosen committee when choosing from $m$ candidates using~$R$.

\subsection{Lower Bound}

We begin by deriving a lower bound analogous to \Cref{thm:singlewinner-lowerbound} from the single-winner setting. Let $h(P, K, \ell; \cdot)$ and $H(P, K, \ell; \cdot)$ be the probability mass function and cumulative distribution function of a hypergeometric random variable with population size $P$, population success size $K$, and draw size $\ell$, respectively.\footnote{Recall that $h(P, K, \ell; t) = \frac{\binom{K}{t} \binom{P - K}{\ell - t}}{\binom{P}{\ell}}$ for $t\in\{0,1,\dots,\ell\}$, and $H(P, K, \ell; t) = \sum_{t' = 0}^t h(P, K, \ell; t')$ for the same range of $t$.} 
In other words, this is the distribution over the number of ``good'' elements sampled when taking a random subset of size $\ell$ out of $P$ elements where $K$ of the elements are good.

\begin{theorem}\label{thm:multiwinner-lowerbound}
    For any $k$-committee rule $R$, any $s \in (0, 1)$, and any $t \le k$, we have \[\FVRMath(R, s, t, m) \ge H(m, \ceiling{sm}, k; t - 1).\]
    Additionally, there exists a $k$-committee rule $R$ (tailored to this $(s, t)$ combination) for which this is tight.
\end{theorem}

The proof of this theorem (along with other results about multi-winner voting) relies on the following correspondence between the number of candidates approved and the number of $k$-committees which are $t$-approved. 
\begin{lemma}\label{lem:t-approve}
    If a voter approves exactly $\ell$ candidates, then she $t$-approves a $(1 - H(m, \ell, k; t - 1))$-fraction of all $k$-committees.
\end{lemma}
\begin{proof}
    This follows from the definition of the hypergeometric distribution. 
    Define the population to be the $m$ candidates and a ``success'' to be the $\ell$ candidates approved by the voter.
    Then, if we pick a $k$-committee uniformly at random, the probability that the voter approves at most $t - 1$ candidates on the committee is $H(m, \ell, k; t - 1)$.
    Therefore, the voter $t$-approves the committee with probability $1 - H(m, \ell, k; t - 1)$. 
\end{proof}

While the bound in \Cref{thm:multiwinner-lowerbound} may not be the most natural-looking, note that as $m$ grows large, the hypergeometric distribution with parameters $m, \ceiling{sm}, k$ approaches a binomial $\text{Bin}(k, s)$ distribution. Hence, the \FVR bound approaches
\[
\sum_{t' = 0}^{t - 1} \binom{k}{t'} s^{t'}(1 - s)^{k-t'}.
\]
Although still unwieldy, the binomial limit gives us a pathway to gain more intuition on this bound. 
By standard concentration inequalities, if $s \ll t/k$, then this bound is exponentially close to $1$, whereas if $s \gg t/k$, then it is exponentially close to $0$. Further, we can see that it simplifies in some special cases.
Indeed, when $k = 1$ (so $t = 1$), the hypergeometric bound simplifies to $1 - \frac{\ceiling{sm}}{m}$. 
This is always at most $1 -s$ and approaches $1 - s$ as $m$ grows; hence, it coincides with the single-winner bound (\Cref{thm:singlewinner-lowerbound}). 
When we fix $t = 1$ but let $k$ be arbitrary, the limiting binomial bound simplifies to $(1 - s)^k$. 
If we fix $t = k$, it simplifies to $1 - s^k$. 
For finite $m$, the hypergeometric distribution is slightly  more concentrated than the binomial distribution---for $t \ll sk$, the binomial bound is an overestimate, while for $t \gg sk$, it is an underestimate. 
Specifically, for $t = 1$ the hypergeometric bound is always less than the binomial limit, and for $t = k$ it is always more.
This aligns with the fact that when sampling without replacement, getting $0$ success is less likely than with replacement. 
Similarly, getting $k$ (out of $k$) successes is less likely without replacement than with replacement.

\subsection{Simultaneously-Optimal Rules}
We will say that a $k$-committee rule $R$ is \emph{\FVR-optimal} for $(s, t)$ if for all $m$, $\FVRMath(R, s, t, m) = H(m, \ceiling{sm}, k; t-1)$.
As with the single-winner setting, \Cref{thm:multiwinner-lowerbound} implies that for any fixed $(s, t)$, we can find an \FVR-optimal rule for this combination. 
However, what if we wish to attain optimality for multiple choices of $(s,t)$ simultaneously? 
We first show that for a fixed $t$, we can achieve optimality over all $s$ via a simple algorithm. 
Given $k$ and $t$, we construct what we call a \emph{$(k, t)$-expanded instance} as follows. 
The candidates in this expanded instance are all $\binom{m}{k}$ possible committees of size~$k$. 
Each voter approves a committee $W$ exactly when she $t$-approves $W$ in the original instance. 
We have the following.
\begin{theorem}\label{thm:ropt-extended} 
Fix $k$ and $t$.
    Let $R^{k, t}_{\emph{expanded}}$ be the rule that runs $\Ropt$ on the $(k, t)$-expanded instance and selects the winning committee. Then, for all $s \in (0, 1)$, $R^{k, t}_{\emph{expanded}}$ is \FVR-optimal for $(s, t)$.
\end{theorem}
This theorem shows that we can directly generalize our single-winner rule to be optimal (in a certain sense) in the multi-winner setting.
However, a downside of $R^{k, t}_{\text{expanded}}$ is that, in its current implementation, it runs in time exponential in $k$, as we must construct an instance with $\binom{m}{k}$ candidates. 
Hence, a natural question is whether we can obtain the same guarantees with a polynomial-time algorithm.
We show next that this is indeed possible using a greedy algorithm, which effectively applies the method of conditional expectations, presented as \Cref{alg:sequential}.

\begin{algorithm}[htb]
    \caption{$(k, t)$ Sequential Algorithm}
    \label{alg:sequential}
    \begin{algorithmic}
    \State $C \gets \emptyset$
        \For{$j = 1, \ldots, k$}
            \For{$i = 1, \ldots, n$ such that $A_i\setminus C\ne\emptyset$}
                \State $w_i \gets \frac{h(m - j - 1,\, |A_i \setminus C| - 1,\, k - j;\, t - 1 - |A_i \cap C|)}{H(m,\, |A_i|,\, k;\, t - 1)}$
            \EndFor
            \State Let $a_j \in M \setminus C$ be a candidate maximizing $\sum_{i \in N_a} w_i$, breaking ties arbitrarily.
            \State $C \gets C \cup \set{a_j}$
        \EndFor
        \Return $C$
    \end{algorithmic}
\end{algorithm}

At each step, the algorithm chooses a candidate that maximizes a certain score among the remaining candidates, and adds it to the committee.
However, unlike in our previous scoring rules, the weight given to each voter $i$ now depends not only on $i$'s flexibility, but also on the size of the current committee and the number of candidates in that committee approved by~$i$. 
The algorithm yields the following guarantee.

\begin{theorem}\label{thm:sequential}
    Fix $k$ and $t$. Let $R^{k, t}_{\emph{alg}}$ be the rule induced by \Cref{alg:sequential}. Then, for all $s \in (0, 1)$, $R^{k, t}_{\emph{alg}}$ is \FVR-optimal for $(s, t)$.
\end{theorem}

To gain some intuition about the algorithm, note that for the $(k, t)$-expanded instance in \Cref{thm:ropt-extended} (and in the proof of \Cref{thm:singlewinner-optimal}), we need not choose a committee \emph{maximizing} the $w$-score (for the optimal $w(f) = \frac{1}{1 - f}$); instead, it is sufficient to choose one with \emph{above-average} score.  
In other words, if we pick a committee $W$ uniformly at random,  $\mathbb{E}[\sc^w(W)]$ would be sufficient for optimal guarantees. 
However, a priori, it is possible that only a single committee $W^*$ has $\sc(W^*) \ge \mathbb{E}[\sc(W)]$ while all others are strictly below, where we drop the superscript $w$ for convenience. 
To get around this and ``round'' the random committee into a deterministic one, we greedily construct $C = \set{a_1, \ldots, a_k}$ such that at each step, $\mathbb{E}[\sc(W) \mid \set{a_1, \ldots, a_j} \subseteq W] \ge \mathbb{E}[\sc(W)]$. 
The weighting rule we use essentially tells us the marginal gain to the score of permanently adding each possible candidate $a$. 
On average, adding each does not change the expectation, so adding one with the largest marginal gain can only improve the conditional expectation, thereby giving the desired bound.

   Next, we show that unlike the positive results we found for simultaneous optimality over $s$ in \Cref{thm:ropt-extended,thm:sequential}, the same does not hold for $t$.
   
   \begin{theorem}\label{thm:impossibility-simultaneous-t} For any $k \ge 2$ and $t\le\lfloor k/2\rfloor$, no rule is simultaneously FVR-optimal for $(1/2,t)$ and $(1/2,k)$.
   \end{theorem}

\subsection{Compatibility with Other Notions of Representation}

Finally, we address the compatibility between \FVR and other notions of representation. 
While several such notions have been proposed in the literature \citep{ABCE+17,SELF+17,PPS21,BP23}, one of the weakest is \emph{justified representation (JR)}.
We show that \FVR-optimality is incompatible even with JR; this implies a similar incompatibility with stronger representation notions (see the discussion in \Cref{sec:additional-discussion}).

 Recall that a committee $W$ of size $k$ is said to satisfy JR if there is no \emph{blocking coalition} $T$ of voters which fulfills the following properties:
\begin{enumerate}
    \item \emph{Large}: $|T| \ge n/k$;
    \item \emph{Cohesive}: $\bigcap_{i \in T} A_i \ne \emptyset$;
    \item \emph{Unrepresented}: $W \cap A_i = \emptyset$ for all $i \in T$.
\end{enumerate}
We establish the following incompatibility.
\begin{theorem}\label{thm:jr-impossibility}
For each $k > 1$, there exists $s$ such that no $k$-committee rule is \FVR-optimal for $(s,1)$ and satisfies JR at the same time.
\end{theorem}

\section{Discussion}

In this paper, we have introduced the notion of voter flexibility and a corresponding objective, 
\emph{\FVRFull (\FVR)}, under approval voting. 
In the single-winner setting, we present a simple rule that is simultaneously optimal for all flexibility thresholds.
In the multi-winner setting, while some impossibilities exist, we still find a polynomial-time rule which makes similar guarantees, once the definition of ``approving a committee'' (i.e., the parameter $t$) is fixed.

One may wonder what the practical implications of these objectives are. 
That is, should we care more about voters who are flexible? 
The answer certainly depends on the context. 
For example, consider an election in which Party A runs $20$ candidates whereas Party B runs only $2$. 
In this setting, Party~A voters are not necessarily ``more flexible'' than Party~B voters, and it may be preferable to simply take the approval winner.
On the other hand, consider a setting where a community is deciding on one of several projects to fund. 
In this case, we may wish to reward voters who are more flexible. 
Indeed, the flexibility of such voters may indicate that they find their disapproved projects strongly unacceptable.
Furthermore, encouraging flexibility can help achieve consensus, an inherently difficult task if few voters are flexible and no candidate has a reasonably-sized base.

More broadly, we believe that flexibility is an interesting foundational metric in its own right. 
Even when the chosen rules are not \FVR-optimal, \FVRFull can be used to compare various candidates and weigh them according to their total support. 
Applied in a way suitable to the context, flexibility measures can help ensure a more desirable outcome in approval-based voting scenarios.

\section*{Acknowledgments}

This work was partially supported by the National Science Foundation under grants IIS-2147187, IIS-2229881, and CCF-2007080, by the Office of Naval Research under grant N00014-20-1-2488, by the Singapore Ministry of Education under grant number MOE-T2EP20221-0001, and by an NUS Start-up Grant. 
We thank the anonymous reviewers for their valuable feedback.

\bibliographystyle{named}
\bibliography{abb,ijcai25}

\appendix

\onecolumn

\vspace{2mm}

\section{Omitted Proofs}
\label{app:omitted}

\subsection{Proof of \Cref{thm:power}}

Fix $p > 0$ and $s\in (0,1)$, and let $g(s,p) := \frac{1}{1 + \frac{(s(1+p))^{1+p}}{p^p}}$.
    
    First, we show that $\FVRMath(R_{p\text{-power}},s) \le g(s,p)$.   
    Take an arbitrary instance, and let $a^*$ be the candidate chosen by $R_{p\text{-power}}$.
    Let $B := \set{i \in N \mid |A_i| \ge sm \text{ and } a^* \notin A_i}$ be the set of $s$-flexible voters who disapprove $a^*$, and let $r := |B|/n$.
    It is sufficient to show that $r \le g(s,p)$.

    Since $N_{a^*} \subseteq N \setminus B$, it holds that $\sc_p(a^*) \le \sum_{i \in N\setminus B} f_i^p$. 
    Let us consider the average score of all candidates, $\frac{1}{m} \sum_{a \in M} \sc_p(a)$. 
    We have that
    \begin{align*}
        \frac{1}{m} \sum_{a \in M} \sc_p(a)
        &= \frac{1}{m} \sum_{a \in M} \sum_{i \in N_a} f_i^p\\
        &= \frac{1}{m^{1+p}}\sum_{a \in M} \sum_{i \in N_a} |A_i|^p\\
        &= \frac{1}{m^{1+p}}\sum_{i \in N} \sum_{a \in A_i}  |A_i|^p
        = \frac{1}{m^{1+p}}\sum_{i \in N}  |A_i|^{1+p}
        = \sum_{i \in N} f_i^{1+p}
        = \sum_{i \in B} f_i^{1+p} + \sum_{i \in N \setminus B} f_i^{1+p}.
    \end{align*}
    
    For each $i \in B$, it holds by definition that $f_i \ge s$, so $\sum_{i \in B} f_i^{1+p} \ge |B| \cdot s^{1+p} = rn s^{1+p}$.
    Since $a^*$ is a candidate with the maximal score, it must hold that $\sc(a^*) 
    \ge \frac{1}{m} \sum_{a \in M} \sc_p(a)$.
    Therefore,
    \begin{align*}
    \sum_{i \in N\setminus B} f_i^p
    \ge \sc(a^*) 
    \ge \frac{1}{m} \sum_{a \in M} \sc_p(a) 
    = \sum_{i \in B} f_i^{1+p} + \sum_{i \in N \setminus B} f_i^{1+p}
    \ge rns^{1+p} + \sum_{i \in N \setminus B} f_i^{1+p}.
    \end{align*}
    That is, $rns^{1+p} \le \sum_{i\in N\setminus B} f_i^p (1-f_i)$.
    Note that within the range $x\in [0,1]$, basic calculus shows that the function $x^p(1-x)$ is maximized at $x = \frac{p}{1+p}$, so we get
    \begin{align*}
    rns^{1+p} 
    \le \sum_{i\in N\setminus B} \left(\frac{p}{1+p}\right)^p\frac{1}{1+p} 
    = \sum_{i\in N\setminus B} \frac{p^p}{(1+p)^{1+p}}
    = (1-r)n \cdot \frac{p^p}{(1+p)^{1+p}}.
    \end{align*}
    Cancelling $n$ from both sides and solving for $r$, we obtain
    \begin{align*}
    r \le \frac{\frac{p^p}{(1+p)^{1+p}}}{s^{1+p} + \frac{p^p}{(1+p)^{1+p}}}
    = \frac{1}{1 + \frac{(s(1+p))^{1+p}}{p^p}} = g(s,p).
    \end{align*}
    This proves that $\FVRMath(R_{p\text{-power}},s) \le g(s,p)$.

    Next, we show that $\FVRMath(R_{p\text{-power}},s) \ge g(s,p)$.
    Take any $r < g(s,p)$.
    It suffices to show that there exists an instance where a group of at least $rn$ voters are all $s$-flexible but do not approve the candidate chosen by the $p$-power scoring rule.
    Consider an instance with sufficiently large $n$ and $m$ (to be made more precise later), and let $a^*$ be a candidate.
    Let $C$ be a group of $\ceiling{rn}$ voters, and assume that each voter in $C$ approves exactly $\ceiling{sm}$ candidates in $M\setminus\{a^*\}$, with the approvals distributed among these candidates as equally as possible.
    Further, assume that each voter in $N\setminus C$ approves $\big\lceil\frac{p}{1+p}\cdot m\big\rceil - 1$ candidates in $M\setminus\{a^*\}$, with the approvals distributed among these candidates as equally as possible; moreover, each of these voters also approves $a^*$.

    In this instance, we have    
    \begin{align*}
    \sc_p(a^*) 
    = \sum_{i \in N\setminus C} f_i^p
    = (n-\ceiling{rn})\cdot\left(\frac{\big\lceil\frac{p}{1+p}\cdot m\big\rceil}{m}\right)^p
    = \floor{(1-r)n}\cdot\left(\frac{\big\lceil\frac{p}{1+p}\cdot m\big\rceil}{m}\right)^p.
    \end{align*}
    On the other hand, each candidate besides $a^*$ is approved by at most $\left\lceil{\frac{\ceiling{rn}\cdot \ceiling{sm}}{m-1}}\right\rceil$ voters in $C$, and approved by at most $\bigg\lceil{\frac{\floor{(1-r)n}\cdot \big(\big\lceil\frac{p}{1+p}\cdot m\big\rceil-1\big)}{m-1}}\bigg\rceil$ voters in $N\setminus C$.
    Hence, the score of each such candidate $a$ is at most
    \begin{align*}
    \sc_p(a) 
    \le \left\lceil{\frac{\ceiling{rn}\cdot \ceiling{sm}}{m-1}}\right\rceil \cdot\left(\frac{\ceiling{sm}}{m}\right)^p 
    + \Bigg\lceil{\frac{\floor{(1-r)n}\cdot \big(\big\lceil\frac{p}{1+p}\cdot m\big\rceil-1\big)}{m-1}}\Bigg\rceil \cdot \left(\frac{\big\lceil\frac{p}{1+p}\cdot m\big\rceil}{m}\right)^p.
    \end{align*}
    We will show that $\sc_p(a^*) > \sc_p(a)$ when $n$ and $m$ are large enough.
    This is sufficient to complete the proof, because it implies that the rule $R_{p\text{-power}}$ selects the candidate $a^*$, which is approved by none of the voters in $C$.

    We have
    \begin{align*}
    \sc_p(a^*) - \sc_p(a) 
    &\ge \floor{(1-r)n}\cdot\left(\frac{\big\lceil\frac{p}{1+p}\cdot m\big\rceil}{m}\right)^p\\
    &\quad- \left\lceil{\frac{\ceiling{rn}\cdot \ceiling{sm}}{m-1}}\right\rceil \cdot\left(\frac{\ceiling{sm}}{m}\right)^p 
    - \Bigg\lceil{\frac{\floor{(1-r)n}\cdot \big(\big\lceil\frac{p}{1+p}\cdot m\big\rceil-1\big)}{m-1}}\Bigg\rceil \cdot \left(\frac{\big\lceil\frac{p}{1+p}\cdot m\big\rceil}{m}\right)^p.
    \end{align*}
    As $m$ grows, $\frac{\big\lceil\frac{p}{1+p}\cdot m\big\rceil}{m}$ and $\frac{\big\lceil\frac{p}{1+p}\cdot m\big\rceil-1}{m-1}$ converge to $\frac{p}{1+p}$, while $\frac{\ceiling{sm}}{m}$ and $\frac{\ceiling{sm}}{m-1}$ converge to $s$.
    In this case, the right-hand side of the inequality above becomes
    \begin{align*}
    \floor{(1-r)n}\cdot\left(\frac{p}{1+p}\right)^p  
    - \ceiling{\ceiling{rn}s}\cdot s^p 
    - \left\lceil\floor{(1-r)n}\cdot \frac{p}{1+p}\right\rceil \cdot \left(\frac{p}{1+p}\right)^p.
    \end{align*}
    It suffices to show that this expression is positive when $n$ is large enough.
    Note that the expression is at least
    \begin{align*}
    &((1-r)n - 1)\cdot\left(\frac{p}{1+p}\right)^p 
    - (((rn+1)s)+1)\cdot s^p
    - \left((1-r)n\cdot \frac{p}{1+p} + 1\right) \cdot \left(\frac{p}{1+p}\right)^p \\
    &= n\left[(1-r)\left(\frac{p}{1+p}\right)^p - rs^{1+p} - (1-r)\left(\frac{p}{1+p}\right)^{p+1}\right]
    - \left(\frac{p}{1+p}\right)^p - (s+1)s^p - \left(\frac{p}{1+p}\right)^p\\
    &= n\left[(1-r)\cdot\frac{p^p}{(1+p)^{1+p}}-rs^{1+p}\right]
    - 2\left(\frac{p}{1+p}\right)^p - (s+1)s^p.    
    \end{align*}
    It remains to show that the expression inside the brackets (attached to $n$) is positive.
    This expression is positive if and only if $r\left(\frac{p^p}{(1+p)^{1+p}} + s^{1+p}\right) < \frac{p^p}{(1+p)^{1+p}}$.
    The inequality holds exactly when 
    \begin{align*}
    r < \frac{\frac{p^p}{(1+p)^{1+p}}}{\frac{p^p}{(1+p)^{1+p}} + s^{1+p}}
    = \frac{1}{1 + \frac{(s(1+p))^{1+p}}{p^p}} = g(s,p),
    \end{align*}
    which is true by assumption.
    This completes the proof. \qed
    
\subsection{Missing Portion of Proof of \Cref{thm:general-weight}}\label{app:general-weight}

We now bound $\FVR(R^w, s)$ from below. 
For $f, f' \in \mathbb{Q}\cap (0, 1)$ such that $w(f) > 0$, let
\[
g(f, f') := \frac{(1 - f) \cdot w(f)}{(1 - f) \cdot w(f) + f' \cdot w(f') }.
\]
Note that this expression is well-defined and positive because $(1 - f) \cdot w(f)> 0$ and $f' \cdot w(f') \ge 0$.
For each pair $f, f'$, we will construct a family of instances such that the fraction of $f'$-flexible voters who disapprove the candidate chosen by $R^w$ becomes arbitrarily close to $g(f, f')$. 
By definition of $\rho$ and $\varphi_s$, there are choices of $f$ and $f' \ge s$ that make this ratio arbitrarily close to
\(
    \frac{\rho}{\rho + \varphi_s},
\)
which establishes the bound. 

Fix such a pair $f,f'$. 
Let $m$ be such that both $\ell := mf$ and $\ell' := mf'$ are integers. 
We will have a special candidate $a^*$, and let $B$ be a set of voters of size $\floor{g(f, f') \cdot n} - m$ (with $n$ sufficiently large that this is nonnegative). 
We will ensure that $a^*$ is the unique candidate with the highest $w$-score, while voters in $B$ are $f'$-flexible and disapprove $a^*$. 
Note that as $n$ grows large, $|B|/n$ approaches the desired bound of $g(f, f')$.

Assume that each voter in $B$ approves $\ell'$ candidates from $M \setminus \set{a^*}$, with the approvals distributed as equally as possible.
Similarly, the voters in $\bar{B} = N\setminus B$ approve $a^*$ along with $\ell - 1$ candidates of $M \setminus \set{a^*}$, again distributed as equally as possible (note that $\ell \ge 1$ because $f > 0$ by assumption). 
Since the voters in $B$ and $\bar{B}$ distribute their approvals as evenly as possible across the candidates in $M \setminus \set{a^*}$, for any such candidate $a$, its score is no more than $w(f') + w(f)$ larger than the average.
Formally,
\[
	\sc^{w}(a) \le \frac{1}{m - 1} \sum_{c \in M \setminus \{a^*\}} \sc^{w}(c) + (w(f') + w(f)).
\]
Hence, it is sufficient to show that
\[
	\sc^{w}(a^*) > \frac{1}{m - 1} \sum_{c \in M \setminus \{a^*\}} \sc^{w}(c) + (w(f') + w(f)).
\]
This is equivalent to showing that
\[
	(m - 1) \sc^{w}(a^*) > \sum_{c \in M \setminus \{a^*\}} \sc^{w}(c) + (m - 1)(w(f') + w(f)).
\]
We have 
\begin{align*}
	(m - 1) \sc^{w}(a^*)
    &= (m - 1) |\bar{B}|\cdot w(f)\\
	&= (\ell - 1) |\bar{B}|\cdot w(f) + (m - \ell)|\bar{B}|\cdot w(f)\\
	&\ge (\ell - 1) |\bar{B}|\cdot w(f) + (m - \ell)\left((1 - g(f, f')) \cdot n + m\right)\cdot w(f)\\
	&= (\ell - 1) |\bar{B}|\cdot w(f) + m \cdot (1 - f)\cdot w(f) \cdot \frac{f' \cdot w(f')}{(1 - f) \cdot w(f) + f' \cdot w(f')} \cdot n  + m \cdot (m - \ell)\cdot w(f)\\
 &= (\ell - 1) |\bar{B}|\cdot w(f) + \ell' \cdot w(f') \cdot g(f, f') \cdot n  + m\cdot (m - \ell) \cdot w(f)\\
 &\ge (\ell - 1) |\bar{B}|\cdot w(f) + \ell' \cdot w(f') \cdot (|B| + m  )  + m\cdot (m - \ell)\cdot w(f)\\
	&= (\ell - 1) |\bar{B}|\cdot w(f) +  w(f') \cdot \ell' \cdot  |B|   + m((m - \ell)\cdot w(f) + \ell \cdot w(f'))\\
	&\ge  (\ell - 1) |\bar{B}|\cdot w(f) +  w(f') \cdot \ell' \cdot  |B|   + m(w(f) + w(f'))\\
    &>  (\ell - 1) |\bar{B}|\cdot w(f) +  w(f') \cdot \ell' \cdot  |B|   + (m - 1)(w(f) + w(f'))\\
    &= \sum_{c \in M \setminus \{a^*\}} \sc^{w}(c) + (m - 1)(w(f') + w(f)),
\end{align*}
where the second-to-last inequality holds because $1\le \ell \le m-1$, and the last inequality holds because $w(f) > 0$.
This completes the proof.
\qed

\subsection{Proof of \Cref{thm:characterization}}

    Fix $w$ and suppose that $R^w$ is \FVR-optimal. 
    By \Cref{thm:general-weight}, this implies that $\frac{\rho}{\rho + \varphi_s} = 1 - s$ for all $s\in (0,1)$, where $\rho = \sup_f (1 - f) \cdot w(f)$ and $\varphi_s = \inf_{f \ge s} f \cdot w(f)$.
    
    First, we show that $w(f) > 0$ for every $f$. 
    Indeed, since $w$ is nontrivial, there exists $f^* \in \mathbb{Q}\cap (0,1)$ for which $w(f^*) > 0$, so we have $\rho \ge w(f^*) \cdot (1 - f^*) > 0$. 
    Further, if $w(f) = 0$ for some $f$, then $\varphi_f \le w(f) \cdot (1 - f) = 0$. Therefore, 
    \[
        \frac{\rho}{\rho + \varphi_f} \ge 1 > 1 - f,
    \]
    a contradiction.
    
    We next show that for all flexibilities $f_1, f_2$, it holds that $(1 - f_2) \cdot w(f_2) = (1 - f_1) \cdot w(f_1)$. 
    Fix $f_1, f_2 \in \mathbb{Q} \cap (0, 1)$. 
    By setting $s = f_1$, we have that
    \[
        \frac{\rho}{\rho + \varphi_{f_1}} = 1 - f_1.
    \]
    Note that the ratio $\frac{a}{a + b}$ is increasing in $a$ and decreasing in $b$. Therefore, replacing $\rho$ by a value at most $\rho$ and replacing $\varphi_{f_1}$ with a value larger than $\varphi_{f_1}$ will only decrease the ratio. In particular, by their definitions, $\rho \ge (1 - f_2) \cdot w(f_2)$ and $\varphi_{f_1} \le f_1 \cdot w(f_1)$. Hence,
    \[
        \frac{(1 - f_2) \cdot w(f_2)}{(1 - f_2) \cdot w(f_2) + f_1 \cdot w(f_1)} \le 1 - f_1.
    \]
    As shown before, each $w(f)$ is positive, so the denominator is positive. Multiplying on both sides, we obtain
    \[
       (1 - f_2) \cdot w(f_2) \le (1 - f_1) \cdot ((1 - f_2) \cdot w(f_2) + f_1 \cdot w(f_1)).
    \]
    Distributing and rearranging yields
    \[
        f_1 \cdot (1 - f_2) \cdot w(f_2) \le f_1 \cdot (1 - f_1) \cdot w(f_1).
    \]
    Since $f_1 > 0$, we can divide both sides by $f_1$ to get
    \[
        (1 - f_2) \cdot w(f_2) \le (1 - f_1) \cdot w(f_1).
    \]
    Redoing the entire argument with $f_1$ and $f_2$ flipped yields the reverse inequality, and hence equality.

    With this in hand, let $c := w(1/2) \cdot (1 - 1/2)$. By plugging in $f$ for $f_1$ and $1/2$ for $f_2$, we immediately get that $w(f) = \frac{c}{1 - f}$ for all $f$. \qed

\subsection{Proof of \Cref{thm:multiwinner-lowerbound}}
    Fix $m$, $k$, $s$, and $t$.
    Consider an instance with $n = \binom{m}{\ceiling{sm}}$ voters, each approving a distinct subset of $\ceiling{sm}$ candidates.
    Note that by symmetry, every $k$-committee $W$ is $t$-disapproved by exactly the same number of voters. 
    By \Cref{lem:t-approve} and an averaging argument, this must be a $H(m, \ceiling{sm}, k; t - 1)$-fraction of voters.
    Hence, no matter which committee is selected, an $H(m, \ceiling{sm}, k; t - 1)$-fraction of voters, all of whom are $s$-flexible, will $t$-disapprove. 
    
    To show tightness, consider a rule that chooses a $k$-committee that is $t$-approved by the highest number of $s$-flexible voters. 
    Note that each $s$-flexible voter approves at least a $(1 - H(m, \ceiling{sm}, k; t - 1))$-fraction of the committees (as approving more than $\ceiling{sm}$ candidates can only increase this value). 
    Hence, by an averaging argument, there must exist a committee approved by at least a $(1-H(m, \ceiling{sm}, k; t - 1))$-fraction of these voters. Therefore, at most an $H(m, \ceiling{sm}, k; t - 1)$-fraction of $s$-flexible voters (and therefore of all voters as well) can $t$-disapprove the chosen committee.
\qed

\subsection{Proof of \Cref{thm:ropt-extended}}
    Fix $k$, $t$, $s$, and an instance with $m$ candidates.
    Let $W$ be the committee chosen by $R^{k, t}_{\text{expanded}}$, and let $B$ be a subset of voters that are $s$-flexible and do not $t$-approve $W$. 
    By \Cref{lem:t-approve}, if a voter is $s$-flexible, then she approves at least a $(1- H(m, \ceiling{sm}, k; t - 1))$-fraction of candidates in the $(k, t)$-expanded instance.
    In other words, the voter is ($1- H(m, \ceiling{sm}, k; t - 1)$)-flexible in the instance on which $\Ropt$ is run. 
    Therefore, \Cref{thm:singlewinner-optimal} guarantees that $|B| \le  H(m, \ceiling{sm}, k; t - 1) \cdot n$.
\qed

\subsection{Proof of \Cref{thm:sequential}}
    Fix $k$, $t$, $m$, and $s$. 
    By \Cref{lem:t-approve}, each voter~$i\in N$ $t$-approves a $(1 - H(m, |A_i|, k; t - 1))$-fraction of $k$-committees. 
    For each $k$-committee $W$, let
    $$\sc(W) :=  \sum_{i : |A_i \cap W| < t} \frac{1}{H(m, |A_i|, k; t- 1)};$$
    note that the sum is taken over all voters who $t$-disapprove $W$.
    Observe that had we run $\Ropt$ on the $(k, t)$-expanded instance, it would choose a committee with the lowest such score. 
    However, a careful inspection of the proof for $\Ropt$ reveals that we do not need to choose a lowest-score committee for the guarantees to hold. 
    Instead, \emph{any} committee with score no greater than the average score already fulfills the guarantee; the minimum is only a convenient choice to ensure this. 
    Formalizing this observation in the multi-winner setting, where the following expectation is taken over all $k$-committees, we have that
    \begin{align*}
        \E[\sc(W)] &= \frac{1}{\binom{m}{k}}\sum_W \sc(W)\\
        &= \frac{1}{\binom{m}{k}} \sum_i \sum_{W: |A_i \cap W| < t} \frac{1}{H(m, |A_i|, k; t- 1)}\\
        &= \frac{1}{\binom{m}{k}} \sum_i \binom{m}{k} \cdot H(m, |A_i|, k; t- 1) \cdot  \frac{1}{H(m, |A_i|, k; t- 1)} = n.
    \end{align*}
    As long as $W$ is such that $\sc(W) \le n$, we claim that it yields \FVR-optimality for $(s,t)$. 
    Indeed, suppose that $B$ is a set of $s$-flexible voters that $t$-disapprove $W$. 
    Then, each voter $i \in B$ contributes $\frac{1}{H(m, |A_i|, k; t- 1)} \ge \frac{1}{H(m, \ceiling{sm}, k; t- 1)}$ points to $\sc(W)$. 
    Since $\sc(W) \le n$, we have $|B| \le n \cdot H(m, \ceiling{sm}, k; t- 1)$.

    It therefore suffices to show that \Cref{alg:sequential} chooses a committee $W$ with $\sc(W) \le n$. 
    Let $a_1, \ldots, a_k$ be the sequence of candidates added to the committee in the $k$ iterations of the loop. 
    For $0\le j\le k$, let $C_j = \set{a_1, \ldots, a_j}$ be the set of the first $j$ candidates, so $C_k$ is the final committee returned by the algorithm. 
    We will show that for each $j\ge 1$, it holds that
    \begin{equation}\label{eq:greedy}
    	a_j \in \argmin_{a \in M \setminus C_{j-1}} \E[\sc(W) \mid C_{j - 1} \cup \set{a} \subseteq W].
    \end{equation}
     In other words, $a_j$ is chosen greedily to minimize $\E[\sc(W) \mid C_j \subseteq W]$. 
     Note that \eqref{eq:greedy} implies that for each $j\ge 1$,
     \[
     	\E[\sc(W) \mid C_j \subseteq W] \le \E[\sc(W) \mid C_{j-1} \subseteq W].
     \]
     Therefore,
     \[
     	\sc(C_k) = \E[\sc(W) \mid C_k \subseteq W] \le \E[\sc(W) \mid C_0 \subseteq W] = \E[\sc(W)] = n,
     \]
     where the first equality holds because $|C_k| = k$ (so conditioning on $C_k \subseteq W$ implies $W = C_k$ with probability $1$), the inequality by a straightforward induction, and the second equality because $C_0 = \emptyset$ (so the conditioning is vacuous).
     
     To complete the proof, it remains to establish \eqref{eq:greedy}. 
     Fix $a \in M \setminus C_{j-1}$ with $a \ne a_j$. 
     Let $C_j' = C_{j - 1} \cup \set{a}$, i.e., the alternate set had we added $a$ instead of $a_j$. 
     We will show that
     \[
     	\E[\sc(W) \mid C_j \subseteq W] \le \E[\sc(W) \mid C'_j \subseteq W].
     \]
     More specifically, we will show that
     \[
     	\E[\sc(W) \mid C_j \subseteq W] - \E[\sc(W) \mid C'_j \subseteq W] \le 0.
     \]
     Note that the conditional distribution of a uniformly selected $W$ conditioned on some set $T \subseteq W$ is simply uniform over the $\binom{m - |T|}{k - |T|}$ sets containing $T$. 
     Hence,
     \[\E[\sc(W) \mid C_j \subseteq W] - \E[\sc(W) \mid C'_j \subseteq W]= \frac{1}{\binom{m - j}{k - j }}\left(\sum_{W:  C_j \subseteq W} \sc(W) - \sum_{W: C'_j \subseteq W} \sc(W)  \right).
     \]
     Since we only wish to show that this quantity is at most $0$, we can ignore the constant in front and show that the difference of sums is nonpositive. 
     We can remove from both sums all committees $W$ that contain both $a_j$ and $a$ to obtain that the difference is equal to
     \[
     	\sum_{\substack{W:  C_j \subseteq W\\a \notin W}} \sc(W) - \sum_{\substack{W:  C'_j \subseteq W\\ a_j \notin W}} \sc(W).
     \]
     An alternative way to sum over the same sets is to simply sum over all sets $T \subseteq M \setminus (C_{j-1} \cup \set{a_j, a})$ with $|T| = k - j$ and output $C_j \cup T$ or $C'_j \cup T$, respectively. 
     For brevity, we will let $R:=  M \setminus (C_{j-1} \cup \set{a_j, a})$ denote the set of these remaining candidates in neither $C_j$ nor $C'_j$.  
     Hence, the above difference is equal to
     \[
     	\sum_{T} (\sc(C_j \cup T) - \sc(C'_j \cup T)),
     \]
     where the sum is taken over all sets $T\subseteq R$ of size $k-j$.
     Expanding the definition of the score, we get that this is equal to
     \begin{equation}\label{eq:indicators}
		\sum_i \left(\frac{1}{H(m, |A_i|, k; t-1)} \cdot\sum_{T}  (\mathbb{I}[|A_i \cap (C_j \cup T)| < t] - \mathbb{I}[|A_i \cap (C'_j \cup T)| < t])\right).
	\end{equation}
    Note that the sets $A_i \cap (C_j \cup T)$ and $A_i \cap (C'_j \cup T)$ differ in size by at most $1$, and the difference is $1$ when exactly one of $a_j$ and~$a$ is contained in $A_i$. 
    More specifically, either $a_j \in A_i$ and $a \notin A_i$ (in which case the first is larger by $1$) or vice versa (in which case the second is larger by $1$). 
    Hence, the only scenarios in which the difference of indicators is nonzero is when exactly one of $a_j$ and $a$ is in $A_i$, and moreover $|A_i \cap (C_{j-1} \cup T)| = t - 1$; the latter condition can be rewritten as $|A_i \cap T| = t - 1 - |A_i \cap C_{j-1}|$. 
    If $a_j \in A_i$, the difference is $-1$, while if $a \in A_i$, the difference is $1$.
    
    Fix a voter $i$ such that exactly one of $a_j$ and $a$ is contained in $A_i$, and consider the number of sets $T \subseteq R$ of size $k-j$ such that $|A_i \cap T| = t - 1 - |A_i \cap C_{j-1}|$. We claim that it is $\binom{m-j-1}{k-j}$ times $h(m - j - 1,\, |A_i \setminus C_{j-1}| - 1,\, k - j;\, t - 1 - |A_i \cap C_{j-1}|)$. 
    Indeed, it is exactly the probability of choosing a random set $T$ of size $k - j$ from a set $R$ of size $m - j - 1$ (recall that $R = M \setminus (C_{j - 1} \cup \set{a_j, a})$) which contains $|A_i \cap R| = |A_i \setminus (C_{j - 1} \cup \set{a_j, a})| = |A_i \setminus C_{j - 1}| - 1$ elements of $A_i$, and ending up with $|A_i\cap T| = t - 1 - |A_i \cap C_{j - 1}|$  such elements. 
    Putting this together, we have that as long as exactly one of $a_j$ and $a$ is in $A_i$, the entire inner sum is $\binom{m-j-1}{k-j}$ times
    \begin{equation}\label{eq:indicators-2}
    	(\mathbb{I}[a \in A_i] - \mathbb{I}[a_j \in A_i]) \cdot h(m - j - 1,\, |A_i \setminus C_{j-1}| - 1,\, k - j;\, t - 1 - |A_i \cap C_{j-1}|).
    \end{equation}
    In fact, \eqref{eq:indicators-2} continues to hold even when not exactly one of $a_j$ and $a$ is in $A_i$, as the difference in indicators is simply $0$ in that case. 
    Hence, \eqref{eq:indicators} is $\binom{m-j-1}{k-j}$ times
    \[
    	\sum_i \left((\mathbb{I}[a \in A_i] - \mathbb{I}[a_j \in A_i])\cdot \frac{h(m - j - 1,\, |A_i \setminus C_{j-1}| - 1,\, k - j;\, t - 1 - |A_i \cap C_{j-1}|)}{H(m, |A_i|, k; t-1)}\right).
    \]
    Note that the probability mass function and cumulative distribution function ratio is exactly the weight $w_i$ given to voter $i$ in the round where $a_j$ is chosen in \Cref{alg:sequential}.
    Rewriting the expression above, we see that this is equal to
    \[
    	\sum_{i: a \in A_i} w_i - \sum_{i: a_j \in A_i} w_i.
    \]
    Since $a_j$ is a candidate with the highest weighted approval when it is chosen, this difference is at most $0$, as desired. \qed

\subsection{Proof of \Cref{thm:impossibility-simultaneous-t}}

Consider an instance with $m = 2k$ candidates partitioned into two sets $C_1$ and $C_2$, each of size $k$. 
Suppose that half of the voters approve only $C_1$, and the other half approve only $C_2$. 
Note that all voters are $1/2$-flexible.

Fix any rule $R$. If $R$ is FVR-optimal for $(1/2, k)$, the chosen committee must be $k$-approved by a nonzero fraction of voters, i.e., it must consist entirely of either $C_1$ or $C_2$. 
However, this means that $R$ cannot be FVR-optimal for $(1/2, t)$. 
Indeed, the other half of the voters do not $t$-approve this committee (and are $1/2$-flexible), so $\text{FVR}(R, 1/2, t, m) \ge 1/2$.
On the other hand, the optimal value is $H(2k, k, k; t - 1) < 1/2$. 
Indeed, it is precisely the probability that these $1/2$-flexible voters do not $t$-approve a randomly chosen committee of size $k$. 
Since $t \le \lfloor k / 2 \rfloor$, by symmetry, this probability is strictly less than $1/2$. 
\qed

\subsection{Proof of \Cref{thm:jr-impossibility}}
    Take some $m > k$. 
    Suppose that the candidates are partitioned into two sets $B \cup C$, where $B = \set{b_1, \ldots, b_{k-1}}$ and $C = \set{c_1, \ldots, c_{m - k + 1}}$. 
    There are $k$ groups of voters, $N_1, \ldots, N_k$, each of size $m - k + 1$.
    For $1\le j\le k-1$, each voter in $N_j$ approves $b_j$ and no other candidate.
    Each voter in $N_k$ approves all candidates in $C$ except one, a distinct one per voter. 
       
    Observe that for a committee $W$ to satisfy JR, we must have $B \subseteq W$. 
    Indeed, if not, there is some $j$ such that $b_j \notin W$, and the set $N_j$ forms a blocking coalition. 
    However, once $B \subseteq W$, we have $|W \cap C| = 1$. 
    Hence, there exists a voter in $N_k$ who disapproves the unique candidate from $C$ in $W$, and therefore disapproves the entire committee. 
    This voter is $\frac{m - k}{m}$-flexible and comprises a $\frac{1}{k \cdot (m - k + 1)} \ge \frac{1}{km}$ fraction of the voters. 
    Note that the optimal \FVR guarantee for $s = \frac{m-k}{m}$ and $t = 1$ is $H(m, m-k, k; 0) = \frac{1}{\binom{m}{k}}$. 
    For sufficiently large $m$, we have $\binom{m}{k} > km$, so the optimal \FVR guarantee for $(s,t) = (\frac{m-k}{m},1)$ is incompatible with JR.
\qed

\section{Proportional Veto Core}
\label{app:PVC}

In this appendix, we define a stronger version of proportional veto core in the same spirit as our FVR notion, and show that it can sometimes be empty.

\begin{definition}
Given $n$ voters with strict rankings over $m$ candidates, a candidate $a$ is \emph{weakly vetoed} by a group $N'$ of voters if each voter in $N'$ prefers at least $m - \lceil m\cdot |N'|/n\rceil + 1$ candidates to $a$.

The \emph{strong proportional veto core} consists of all candidates that are not weakly vetoed by any group of voters.
\end{definition}

Note that in the original proportional veto core \citep{Moul81}, a candidate $a$ is vetoed by a group of voters $N'$ if all voters in $N'$ prefer the \emph{same} set of $m - \lceil m\cdot |N'|/n\rceil + 1$ candidates to $a$.

\begin{proposition}
The strong proportional veto core can be empty.
\end{proposition}

\begin{proof}
Consider $m = 3$ candidates and $n = 2$ voters with rankings $a\succ b\succ c$ and $c\succ b\succ a$.
\begin{itemize}
\item Candidate $c$ is vetoed by voter~$1$, since the voter prefers $3 - \lceil 3\cdot 1/2\rceil + 1 = 2$ candidates to $c$.
\item Similarly, candidate $a$ is vetoed by voter~$2$.
\item Candidate $b$ is vetoed by both voters together, since each voter prefers $3 - \lceil 3\cdot 2/2\rceil + 1 = 1$ candidate to $b$.
\end{itemize}
Hence, the strong proportional veto core is empty in this instance.
\end{proof}

\end{document}